\documentclass[11pt,leqno]{article}

\pdfoutput=1

\usepackage{amsmath,amsthm,amssymb}
\usepackage{authblk,graphicx,subfigure,enumerate,booktabs,color} 
\usepackage{mathabx}

\let\originalincludegraphics\includegraphics
\renewcommand{\includegraphics}[2][]{%
 \IfFileExists{#2}{\originalincludegraphics[#1]{#2}}{%
  \fbox{\texttt{\detokenize{#2}}}}}

\theoremstyle{plain}
\newtheorem{thm}{Theorem}[section]
\newtheorem{prop}[thm]{Proposition}

\newtheorem{lem}[thm]{Lemma}

\theoremstyle{remark}
\newtheorem{rem}[thm]{Remark}

\theoremstyle{definition}
\newtheorem{ex}[thm]{Example}

\newtheorem{assum}[thm]{Assumption}

\numberwithin{equation}{section}
\numberwithin{figure}{section}
\numberwithin{table}{section}

\title{Kernel-based collocation methods for \\ Heath-Jarrow-Morton models with  
Musiela parametrization}

\author{Yuki Kinoshita\thanks{Email: yk.lgb13@gmail.com}}
\author[1]{Yumiharu Nakano\thanks{Corresponding author. Email:
nakano@comp.isct.ac.jp}}

\affil[1]{Department of Mathematical and Computing Science, School of Computing,\\
Institute of Science Tokyo, W8-28, 2-12-1 Ookayama, Meguro-ku,\\
Tokyo 152-8550, Japan}

\date{\today}

\begin{document}

\maketitle

\begin{abstract}
\begin{sloppypar}
We study kernel-based collocation methods for numerical solutions of
Heath--Jarrow--Morton models in Musiela parametrization.  The methods are
Euler--Maruyama approximations of finite-dimensional stochastic differential
equations and can be used with standard Monte Carlo procedures for derivative
pricing.  We give a conditional mean-square error estimate under explicit
stability assumptions on the differentiated cardinal basis over the full
evaluation region required by the volatility functional.  The estimate uses
the spatial-radius exponent dictated by the squared-error calculation.  Two
finite-resolution examples illustrate the method for a Vasicek model and for
caplet pricing.
\end{sloppypar}
\begin{flushleft}
\end{flushleft}
\begin{flushleft}
\end{flushleft}
\end{abstract}



\section{Introduction}\label{sec:1}

In this paper, we are concerned with numerical methods for Heath-Jarrow-Morton (HJM) models 
with Musiela parametrization. Consider the forward rate process $f(t,T)$, 
$0\le t\le T<\infty$, given as a family of It{\^o} processes, 
in an arbitrage-free bond market. 
Under the usual no-arbitrage assumptions, Heath et al.~\cite{hea-etal:1992}
show that $f(t,T)$ must satisfy
\begin{equation}
\label{eq:1.1} 
 df(t,T) = \alpha(t,T)dt + \sum_{i=1}^d\sigma_i(t,T)dW_i(t). 
\end{equation}
Here, this equation is defined on a complete probability space
$(\Omega,\mathcal{F},\mathbb{P})$ with a filtration $\{\mathcal{F}(t)\}_{t\ge 0}$  
satisfying the usual conditions. 
The probability measure $\mathbb{P}$ is interpreted as a pricing measure.
The process $W(t)=(W_1(t),\ldots,W_d(t))$, $t\ge 0$, is a standard $d$-dimensional 
$\{\mathcal{F}(t)\}$-Brownian motion under $\mathbb{P}$. 
The processes $\sigma_i(t,T)$, $i=1,\ldots,d$, are assumed to be
appropriately measurable and integrable, and the process $\alpha(t,T)$ is given by   
\begin{equation*}
 \alpha(t,T)= \sum_{i=1}^{d}\sigma_i(t,T)\int_t^T\sigma_i(t,s)ds. 
\end{equation*}
We refer to standard textbooks such as Musiela and Rutkowski \cite{mus-rut:2007},  
Shreve \cite{shr:2004}, Bj{\"o}rk \cite{bjo:2004} and the references therein 
for details and developments of HJM models \eqref{eq:1.1}. 
Musiela \cite{mus:1993} showed that $r(t,x):=f(t,t+x)$, which is called
the Musiela parametrization, is a mild solution to the stochastic partial differential equation 
\begin{equation}
\label{eq:1.2}
 dr(t,x) = \left(\frac{\partial}{\partial x}r(t,x) + \alpha(t,t+x)\right)dt 
  + \sum_{i=1}^d\sigma_i(t,t+x)dW_i(t)
\end{equation}
in a suitable function space. 
The equation \eqref{eq:1.2} is called the Heath-Jarrow-Morton-Musiela (HJMM) equation. 
Existence and uniqueness for variants of \eqref{eq:1.2} have since been
studied extensively. See, e.g., Goldys and Musiela \cite{gol-mus:2001}, Filipovi{\'c} \cite{fil:2001}, 
Barski and Zabczyk \cite{bar-zab:2012}, 
Kusuoka \cite{kus:2000} and the references 
therein. 

As for numerical methods for \eqref{eq:1.2}, 
Barth \cite{bar:2010} studies a finite element method, and
D{\"o}rsek and Teichmann \cite{dor-tei:2013} propose a splitting method.
In the present paper, we examine kernel-based collocation methods for numerical solutions 
to \eqref{eq:1.2} when $\sigma$ depends on $f(t,\cdot)$, and hence on $r(t,\cdot)$,
as an alternative to existing methods. 

Given a point set
$\Gamma=\{x_1,\ldots,x_N\}$ such that 
$0<x_1<\cdots <x_N$, and a positive definite function  
$\Phi:\mathbb{R}\to\mathbb{R}$, the function  
\[
 I(g)(x) = \sum_{j=1}^N(K^{-1}g|_{\Gamma})_j\Phi(x-x_j), \quad x\in\mathbb{R}, 
\]
interpolates $g$ on $\Gamma$. Here, $K=\{\Phi(x_j -x_{\ell})\}_{j,\ell=1,\ldots,N}$, 
$g|_{\Gamma}$ is the column vector composed of $g(x_j)$, 
$j=1,\ldots,N$, and 
$(K^{-1}g|_{\Gamma})_j$ denotes the $j$-th component of $K^{-1}g|_{\Gamma}\in\mathbb{R}^N$.  
Since one can expect
\[
 \frac{d^m}{dx^m}g(x) \approx \frac{d^m}{dx^m}I(g)(x), \quad m=0,1, 
\]
we replace the curve and its derivative in \eqref{eq:1.2} by their kernel
interpolants.  Collocation at the points in $\Gamma$ then gives an
$N$-dimensional stochastic differential equation.
See Section \ref{sec:3} below for a more precise derivation. 
Methods of this type are generally called kernel-based collocation methods.
They were introduced by Kansa \cite{kan:1986}; see also
\cite{kan:1990a,kan:1990b}.
Since then, kernel-based methods have been studied extensively through
numerical experiments and practical applications.
Rigorous convergence issues are studied in, e.g.,  
Schaback \cite{sch:2010}, Cialenco et al.~\cite{cia-etal:2012},
Hon et al.~\cite{hon-etal:2014},
Nakano \cite{nak:2017, nak:2019a}.
Our aim is to analyze kernel-based collocation methods for HJMM equations
and to obtain an error bound for these methods.
To this end, we state explicitly the stability property of the differentiated
cardinal basis needed by the analysis and develop an interpolation estimate for
a class of functions with relatively low regularity.

\paragraph{Relation to the published version.}
The version published in \emph{Stochastics} 93 (2021), 921--944, derived its
principal convergence theorem from a stability lemma imported from
\cite{nak:2019a}.  A subsequent re-examination of that lemma showed that its
proof does not establish the asserted uniform bound on the absolute row sums of
the differentiated cardinal basis.  In the Kergin-interpolation step, prescribed
sign values were available only at a selected subset of the collocation points,
whereas the proof used a bound for the interpolating polynomial at an arbitrary
nearest collocation point.  The latter point need not belong to the selected
subset.  Consequently, the counting condition in the published Assumption
3.2(ii) does not imply the stability estimate used in the convergence proof.

The present version withdraws that implication and replaces the counting
condition by the required row-sum bound as an explicit assumption.  Its range is
enlarged from the finite set $\Gamma\cup\Gamma_e$ to
$[0,T_0]\cup\Gamma\cup\Gamma_e$, because the volatility Lipschitz condition
involves integrals over $[0,T_0]$.  This also repairs an independent gap in the
published proof, where estimates established only on the finite evaluation set
were applied inside such integrals.  Finally, the published theorem displayed
the spatial factor
$(\Delta x)^{(2\tau-1)/\tau}R^{1/(2\tau)}$, whereas the squared-error estimates
in its proof yield
$(\Delta x)^{(2\tau-1)/\tau}R^{1/\tau}$.  The theorem below uses the latter
factor and imposes its convergence to zero.  Thus the corrected theorem is a
conditional result: it does not establish the new stability assumption for the
specific kernel and point-set sequences used in the numerical section.

The revision also repairs several independent errors in the supporting
arguments.  In the proof of strong continuity of the shift semigroup, the
distance formula used to approximate an indicator function had its numerator
reversed.  In the HJM drift lemma, the displayed second derivative of
$\phi(x)\int_0^x\phi(y)\,dy$ was incorrect; it is replaced here by a bilinear
estimate on the zero-limit subspace of $U$.  In the convergence proof, the
sign of $\Theta_i$ in the definition of $\mathcal I_6$ is corrected, and a
pointwise mean-square bound for $\widetilde S(t)$ is proved explicitly so that
no interchange of $\mathbb E\sup$ and $\sup\mathbb E$ is required.  In addition,
the interpolation estimate is derived directly from Wendland's zeros theorem,
with the scaling used to make the independence from $R$ explicit.  Finally,
the inconsistent time-step notation in the caplet discount factor is replaced
by a general quadrature grid through the payment date.

This paper is organized as follows. The next section is devoted to the proof of the existence and 
uniqueness result for \eqref{eq:1.2} in a Hilbert space that is suitable for our purpose.  
We describe the kernel-based collocation methods in detail and derive
the approximation error in Section \ref{sec:3}.  
Section \ref{sec:4} presents two finite-resolution illustrations.

\section{HJMM equations}\label{sec:2}

We describe Heath-Jarrow-Morton models with Musiela parametrization 
or HJMM equations for 
interest rate modeling in a way suitable for our purpose. 
Our setup is based on \cite{fil:2001}, with a slight modification. 

We first introduce some notation.
Let $\mathbb{R}_+=[0,\infty)$. 
For any open or closed set $V\subset \mathbb{R}$ 
we write $\mathcal{B}(V)$ for the Borel $\sigma$-field of $V$. 
We use $\mathrm{Leb}$ to denote 
the Lebesgue measure on $(\mathbb{R},\mathcal{B}(\mathbb{R}))$.  
We put $L^p(V)=L^p(V,\mathcal{B}(V), \mathrm{Leb})$ 
for $p\in [1,\infty]$ and denote by $\|\cdot\|_{L^p(V)}$ its norm. 
We also denote by 
$L^1_{loc}(\mathbb{R}_+)$
the collection of all Borel measurable and locally Lebesgue integrable functions on 
$\mathbb{R}_+$. 
Denote by $C^k(V)$  
the space of all $C^k$-functions on $V$, and 
by $C^k_b(V)$ the collection of all functions in $C^k(V)$ such that   
\[
 \|u\|_{C^k(V)} := \sum_{m=0}^k\sup_{x\in V}
 \left|\frac{d^mu}{dx^m}(x)\right| < \infty. 
\] 
By $C$ we denote positive constants that may vary at each occurrence and that do not depend on  
time and spatial variables in $\mathbb{R}_+$, elements in $\Omega$ and $U$, 
and the approximation parameter $h$ introduced below. 

We work in the Hilbert space 
\begin{equation*}
 U := \left\{\phi\in L^1_{loc}(\mathbb{R}_+) \;\;\bigg|\;\;
 \begin{aligned}
 &\text{there exist the generalized derivatives } \\
 &\phi^{\prime}, \phi^{\prime\prime}\in L^1_{loc}(\mathbb{R}_+) 
 \text{ of } \phi \text{ such that } \|\phi\|_U<\infty
\end{aligned}
\right\}
\end{equation*}
with the norm $\|\cdot\|_U$ defined by  
\begin{equation*}
 \|\phi\|_U^2=|\phi(0)|^2 + |\phi^{\prime}(0)|^2 
   + \int_0^{\infty}\left(|\phi^{\prime}(x)|^2+|\phi^{\prime\prime}(x)|^2\right)w(x)dx, 
\end{equation*}
where  
$w:\mathbb{R}_+\to [1,\infty)$ is a nondecreasing $C^1$-function such that 
$w^{-1/3}\in L^1(\mathbb{R}_+)$.  

We consider the mapping $S(t):U\to U$ defined by $S(t)\phi(x)=\phi(t+x)$, $t,x\in\mathbb{R}_+$. 
It is clear that $\{S(t)\}_{t\in\mathbb{R}_+}$ defines a semigroup on $U$. 
The following proposition summarizes the properties of this semigroup.
\begin{prop}
\label{prop:2.1}
\begin{enumerate}[\rm (i)]
\item The Hilbert space $U$ is separable and satisfies 
$U\subset C^{1}_b(\mathbb{R}_+)$. In particular, 
\[
 \|\phi\|_{L^{\infty}(\mathbb{R}_+)} + \|\phi^{\prime}\|_{L^{\infty}(\mathbb{R}_+)} 
 + \|\phi^{\prime}\|_{L^1(\mathbb{R}_+)} + \|\phi^{\prime\prime}\|_{L^1(\mathbb{R}_+)} 
\le C \|\phi\|_U. 
\]
\item The semigroup $\{S(t)\}_{t\in\mathbb{R}_+}$ is strongly continuous on $U$, 
and the domain of its generator $A$ is given by $\{\phi\in U: \phi^{\prime}\in U\}$. 
Moreover, $A$ satisfies $A\phi = \phi^{\prime}$.  
\end{enumerate}
\end{prop}
\begin{proof}
First, we show that $U$ is separable. To this end, consider the Hilbert space
\[
 U_1:=\{\phi\in L^1_{loc}(\mathbb{R}_+): \text{there exists } 
 \phi^{\prime}\in L^1_{loc}(\mathbb{R}_+) 
 \text{ of } \phi \text{ such that } \|\phi\|_{U_1}<\infty\}, 
\]
where 
\[
 \|\phi\|_{U_1}^2=|\phi(0)|^2 + \int_0^{\infty}|\phi^{\prime}(x)|^2w(x)dx. 
\]
By Theorem 5.1.1 in \cite{fil:2001}, the space $U_1$ is separable. 
The mapping $\phi\mapsto(\phi,\phi')$ identifies $U$ isometrically with a
closed subspace of $U_1\times U_1$. Thus $U$ is separable.

Since $\phi\in U$ has the generalized derivatives $\phi^{\prime}$ and $\phi^{\prime\prime}$, 
we can write 
\begin{equation}
\label{eq:2.1}
 \phi(x) - \phi(y) = \int_y^x\phi^{\prime}(z)dz, \quad  
 \phi^{\prime}(x) - \phi^{\prime}(y) = \int_y^x\phi^{\prime\prime}(z)dz, \quad x,y\in\mathbb{R}_+. 
\end{equation}
Further, as in the proof of Theorem 5.1.1 in \cite{fil:2001}, we see 
\begin{equation*}
 \|\phi^{\prime}\|_{L^1(\mathbb{R}_+)}\le 
 \left(\int_0^{\infty}|\phi^{\prime}(x)|^2w(x)dx\right)^{1/2}
 \|w^{-1}\|_{L^1(\mathbb{R}_+)}^{1/2}\le C\|\phi\|_U<\infty. 
\end{equation*}
Combining this with \eqref{eq:2.1}, we have 
$\|\phi\|_{L^{\infty}(\mathbb{R}_+)}\le C\|\phi\|_U$. 
Similarly, we see $\|\phi^{\prime\prime}\|_{L^1(\mathbb{R}_+)}\le C \|\phi\|_U$, and so 
$\|\phi^{\prime}\|_{L^{\infty}(\mathbb{R}_+)}\le C \|\phi\|_U$. 
Thus the claim (i) follows. 

It follows directly that $(S(t)\phi)^{\prime}$ and $(S(t)\phi)^{\prime\prime}$ exist
and are given by $(S(t)\phi)^{\prime}(x)=S(t)\phi^{\prime}(x)$ and 
$(S(t)\phi)^{\prime\prime}(x)=S(t)\phi^{\prime\prime}(x)$ for $\phi\in U$. 
Using \eqref{eq:2.1} and the monotonicity of $w$, we find 
\begin{align*}
 \|S(t)\phi\|_U^2&= |\phi(t)|^2 + |\phi^{\prime}(t)|^2 
 + \int_0^{\infty}(|\phi^{\prime}(t+x)|^2+|\phi^{\prime\prime}(t+x)|^2)w(x)dx \\ 
&\le C\|\phi\|_U^2 + \int_0^{\infty}
 (|\phi^{\prime}(t+x)|^2+|\phi^{\prime\prime}(t+x)|^2)w(t+x)dx 
 \le C\|\phi\|_U^2.  
\end{align*}
Thus $S(t)$ is bounded on $U$ for every $t\in\mathbb{R}_+$.
By part (i), to prove the strong continuity of
$\{S(t)\}_{t\in\mathbb{R}_+}$ it suffices to show that, for every
$t_0\in\mathbb{R}_+$ and every Borel measurable function
$g$ on $\mathbb{R}_+$ with $g^2w\in L^1(\mathbb{R}_+)$, 
\begin{equation}
\label{eq:2.1.3}
 \lim_{t\to t_0}\int_0^{\infty}|g(t+x)-g(t_0+x)|^2w(x)dx = 0. 
\end{equation}
To this end, for any $\varepsilon>0$ take  
a bounded $E_{\varepsilon}\in\mathcal{B}(\mathbb{R}_+)$ and 
a continuous function $g_{\varepsilon}$ on $\mathbb{R}_+$ such that 
$g_{\varepsilon}(x)=0$ for $x\notin E_{\varepsilon}$ and that 
\[
 \int_0^{\infty}|g(x)-g_{\varepsilon}(x)|^2w(x)dx < \varepsilon. 
\]
The existence of $E_{\varepsilon}$ and $g_{\varepsilon}$ can be
proved by a routine argument in measure theory, but for completeness 
we give a proof later. 
Suppose at this moment that there exist such $E_{\varepsilon}$ and $g_{\varepsilon}$. 
Then take $\ell>0$ so that $t+E_{\varepsilon}, t_0+E_{\varepsilon} \subset [0,\ell]$ for 
$t\ge 0$ with $|t-t_0|\le 1$. By the monotonicity of $w$, 
\begin{align*}
 &\int_0^{\infty}|g(t+x)-g(t_0+x)|^2w(x)dx \\ 
 &\le 3\int_0^{\infty}|g(t+x)-g_{\varepsilon}(t+x)|^2w(x)dx 
 +3\int_0^{\infty}|g_{\varepsilon}(t+x)-g_{\varepsilon}(t_0+x)|^2w(x)dx \\ 
 &\quad +3\int_0^{\infty}|g_{\varepsilon}(t_0+x)-g(t_0+x)|^2w(x)dx \\ 
&\le 3\int_0^{\infty}|g(t+x)-g_{\varepsilon}(t+x)|^2w(t+x)dx 
 +3\int_0^{\ell}|g_{\varepsilon}(t+x)-g_{\varepsilon}(t_0+x)|^2w(x)dx \\ 
 &\quad +3\int_0^{\infty}|g_{\varepsilon}(t_0+x)-g(t_0+x)|^2w(t_0+x)dx \\ 
 &\le 6\int_0^{\infty}|g(x)-g_{\varepsilon}(x)|^2w(x)dx 
   + 3\sup_{x\in [0,\ell]}|g_{\varepsilon}(t+x)-g_{\varepsilon}(t_0+x)|^2\int_0^{\ell}w(x)dx  
\end{align*}
Thus the uniform continuity of $g_{\varepsilon}$ leads to 
\[
 \limsup_{t\to t_0}\int_0^{\infty}|g(t+x)-g(t_0+x)|^2w(x)dx\le 6\varepsilon, 
\]
whence \eqref{eq:2.1.3}. 

To prove the existence of $E_{\varepsilon}$ and $g_{\varepsilon}$,
it suffices to treat $g\ge0$, since the positive and negative parts can be
approximated separately.
Then there exists a nondecreasing sequence of simple functions $\{g_n\}$ such that 
$g_n$ vanishes outside $[0,n)$ and $g_n\to g$ a.e. By the dominated convergence theorem,
we also have 
\[
 \lim_{n\to\infty}\int_0^{\infty}|g(x)-g_n(x)|^2w(x)dx = 0. 
\]
Fix $n\in\mathbb{N}$ such that 
\[
 \int_0^{\infty}|g(x)-g_{n}(x)|^2w(x)dx <\frac{\varepsilon}{4}. 
\]
Write $g_n=\sum_{j=1}^m\alpha_j1_{E_j}$ with $\alpha_j>0$.
By the absolute continuity of the Lebesgue integral, for each $j=1,\ldots,m$ there exists 
$\delta_j>0$ such that for any $E^{\prime}\in\mathcal{B}([0,n))$ with 
$\mathrm{Leb}(E^{\prime})<\delta_j$ we have 
\[
 \int_{E^{\prime}}w(x)dx < \frac{\varepsilon}{4m^2\alpha_j^2}. 
\]
Now take a closed set $F_j$ and an open set $G_j$ such that 
$F_j\subset E_j\subset G_j\subset [0,n)$ with $\mathrm{Leb}(G_j\setminus F_j)<\delta_j$. 
Define the continuous function $\rho_j$ on $\mathbb{R}_+$ by 
\[
 \rho_j(x) = \frac{\inf_{y\in G_j^c}|x-y|}
 {\inf_{y\in F_j}|x-y| + \inf_{y\in G_j^c}|x-y|}, \quad 
 x\in \mathbb{R}_+. 
\]
By construction,
\[
 \int_0^{\infty}|1_{E_j}(x)-\rho_j(x)|^2w(x)dx \le \int_{G_j\setminus F_j}w(x)dx 
 < \frac{\varepsilon}{4m^2\alpha_j^2}. 
\]
Thus the function $g_{\varepsilon}:=\sum_{j=1}^m\alpha_j\rho_j$ satisfies 
\[
 \int_0^{\infty}|g_n-g_{\varepsilon}|^2w(x)dx \le m\sum_{j=1}^m\alpha_j^2 
 \int_0^{\infty}|1_{E_j}(x)-\rho_j(x)|^2w(x)dx <\frac{\varepsilon}{4}. 
\]
Therefore 
\begin{align*}
 \int_0^{\infty}|g(x) - g_{\varepsilon}(x)|^2w(x)dx 
 &\le 2\int_0^{\infty}|g(x)-g_n(x)|^2w(x)dx + 2\int_0^{\infty}|g_n(x)-g_{\varepsilon}(x)|^2w(x)dx \\ 
 &<\varepsilon, 
\end{align*}
Taking $E_\varepsilon=\bigcup_{j=1}^mG_j$, we also have
$g_\varepsilon=0$ outside the bounded set $E_\varepsilon$, as required.

Now, as in the proof of Corollary 5.1.1 in \cite{fil:2001}, we observe, for $\phi\in U$ with 
$\phi^{\prime}\in U$, 
\begin{align*}
 &\left\|\frac{S(t)\phi-\phi}{t}-\phi^{\prime}\right\|_U^2 \\ 
 &\le \left|\frac{\phi(t)-\phi(0)}{t}-\phi^{\prime}(0)\right|^2 
  + \left|\frac{\phi^{\prime}(t)-\phi^{\prime}(0)}{t}-\phi^{\prime\prime}(0)\right|^2 
  + 2\int_0^1\|S(st)\phi^{\prime}-\phi^{\prime}\|^2_Uds \to 0, 
\end{align*}
as $t\to 0$. Hence $A\phi=\phi^{\prime}$. 
Moreover, by the claim (i), the pointwise evaluation operator is continuous. 
This together with the strong continuity of $S$ means that the domain of 
$A$ is included in $\{\phi\in U: \phi^{\prime}\in U\}$ (see Lemma 4.2.2 in \cite{fil:2001}). 
Thus the claim (ii) follows. 
\end{proof}

Let $\sigma_i$, $i=1,\ldots,d$, be measurable mappings from 
$(\mathbb{R}_+\times\Omega\times U, \mathcal{P}\otimes\mathcal{B}(U))$ into 
$(U,\mathcal{B}(U))$, 
where $\mathcal{P}$ denotes the predictable $\sigma$-field 
and $\mathcal{B}(U)$ is the Borel $\sigma$-field of $U$, 
such that $\lim_{x\to\infty}\sigma_i(t,\omega,\phi)(x)=0$ for every 
$i=1,\ldots,d$, $t\in\mathbb{R}_+$, $\omega\in\Omega$, and $\phi\in U$. 
Assume the following uniform bounds.
\begin{assum}
\label{assum:2.2}
There exists a constant $C_1\in (0,\infty)$ such that for $i=1,\ldots,d$ and   
$(t,\omega,\phi,\psi)\in \mathbb{R}_+\times\Omega\times U\times U$, 
\begin{gather*}
 \|\sigma_i(t,\omega,\phi)\|_{U}\le C_1, \\
 \|\sigma_i(t,\omega,\phi)-\sigma_i(t,\omega,\psi)\|_{U}\le C_1\|\phi-\psi\|_U.
\end{gather*}
\end{assum}

When no confusion can arise, we suppress the variable $\omega$.
Define the mapping $\alpha$ on $\mathbb{R}_+\times\Omega\times U$ by 
\begin{equation*}
 \alpha(t,\omega,\phi)(x):= \sum_{j=1}^{d}\sigma_j(t,\omega,\phi)(x)
 \int_0^x\sigma_j(t,\omega,\phi)(y)dy,
  \quad x\in\mathbb{R}_+. 
\end{equation*}
The following lemma records the resulting properties of $\alpha$.
\begin{lem} 
Under Assumption {\rm \ref{assum:2.2}}, 
the mapping $\alpha$ is measurable from 
$(\mathbb{R}_+\times\Omega\times U,\mathcal{P}\otimes\mathcal{B}(U))$ into 
$(U,\mathcal{B}(U))$. Moreover, there exists a constant $C_2\in (0,\infty)$ such that
for $(t,\omega,\phi,\psi)\in\mathbb{R}_+\times\Omega\times U\times U$,  
\begin{gather*}
  \|\alpha(t,\omega,\phi)\|_U\le C_2, \\
  \|\alpha(t,\omega,\phi)-\alpha(t,\omega,\psi)\|_U\le C_2\|\phi-\psi\|_U. 
\end{gather*}
\end{lem}
\begin{proof}
Set
\[
 U_0=\{\phi\in U:\lim_{x\to\infty}\phi(x)=0\}.
\]
The limit exists by Proposition \ref{prop:2.1}, and $U_0$ is a closed
subspace of $U$ because
$\phi(\infty)=\phi(0)+\int_0^\infty\phi'(x)dx$ is a continuous linear
functional on $U$.
For $\phi,\psi\in U_0$, define
\[
 \mathcal{B}(\phi,\psi)(x)=\phi(x)\int_0^x\psi(y)dy,
 \quad x\in\mathbb{R}_+.
\]
The proof of Theorem 5.1.1 in \cite{fil:2001} gives
\[
 \|(\phi-\phi(\infty))^4w\|_{L^1(\mathbb{R}_+)}
 \le C\|\phi\|_U^4, \quad \phi\in U.
\]
Hence H\"older's inequality and $w^{-1/3}\in L^1(\mathbb{R}_+)$ imply
\[
 \|\phi\|_{L^1(\mathbb{R}_+)}
 \le \left(\int_0^\infty |\phi(x)|^4w(x)dx\right)^{1/4}
      \left(\int_0^\infty w(x)^{-1/3}dx\right)^{3/4}
 \le C\|\phi\|_U, \quad \phi\in U_0.
\]
Moreover,
\begin{align*}
 \mathcal{B}(\phi,\psi)'(x)
 &=\phi'(x)\int_0^x\psi(y)dy+\phi(x)\psi(x),\\
 \mathcal{B}(\phi,\psi)''(x)
 &=\phi''(x)\int_0^x\psi(y)dy
   +2\phi'(x)\psi(x)+\phi(x)\psi'(x).
\end{align*}
Proposition \ref{prop:2.1}, the preceding $L^1$ estimate, and the
weighted $L^4$ estimate show that
\[
 \|\mathcal{B}(\phi,\psi)\|_U
 \le C\|\phi\|_U\|\psi\|_U,
 \qquad \phi,\psi\in U_0.
\]
For completeness, the product term in the first derivative is controlled by
\[
 \int_0^\infty |\phi(x)\psi(x)|^2w(x)dx
 \le \left(\int_0^\infty |\phi(x)|^4w(x)dx\right)^{1/2}
      \left(\int_0^\infty |\psi(x)|^4w(x)dx\right)^{1/2};
\]
all other terms follow directly from the $L^\infty$, weighted $L^2$, and
$L^1$ bounds above.  Set $\mathcal{S}\phi=\mathcal{B}(\phi,\phi)$.  Then
\[
 \|\mathcal{S}\phi\|_U\le C\|\phi\|_U^2,
 \qquad
 \|\mathcal{S}\phi-\mathcal{S}\psi\|_U
 \le C(\|\phi\|_U+\|\psi\|_U)\|\phi-\psi\|_U
\]
for $\phi,\psi\in U_0$.
These estimates and the boundedness of $\sigma_i$ yield, for $\phi\in U$,
\[
 \|\alpha(t,\phi)\|_U\le \sum_{i=1}^d\|\mathcal{S}\sigma_i(t,\phi)\|_U
 \le C\sum_{i=1}^d\|\sigma_i(t,\phi)\|_U^2\le C. 
\]
Because $\mathcal{S}$ is continuous on $U_0$ and each $\sigma_i$ is
measurable and $U_0$-valued, $\alpha$ is measurable and $U$-valued.  The
second estimate for $\mathcal{S}$, together with Assumption
\ref{assum:2.2}, gives the asserted Lipschitz bound for $\alpha$.
\end{proof}

By Theorem 7.2 in Da Prato and Zabczyk \cite{dap-zab:2014}, for a given $r_0\in U$,
there exists a unique $U$-valued predictable process $r(t)=r(t,\cdot)$, $t\in\mathbb{R}_+$,  
that is a mild solution to 
\begin{equation}
\label{eq:2.2}
\left\{
\begin{aligned}
 dr(t) &= (Ar(t) + \alpha(t,r(t)))dt + \sum_{i=1}^{d}\sigma_i(t,r(t))dW_i(t), \\
 r(0)&=r_0. 
\end{aligned}
\right.
\end{equation}
Moreover, $\{r(t)\}_{t\ge 0}$ has a continuous modification and satisfies
\begin{equation}
\label{eq:2.3}
\sup_{0\le t\le T}\mathbb{E}\|r(t)\|_U^2\le C_T(1+\|r_0\|_U^2)
\end{equation}
for some positive constant $C_T$, for every $T>0$.
Therefore, for $t\in\mathbb{R}_+$, 
\begin{equation}
\label{eq:2.4}
 r(t) = S(t)r_0 + \int_0^tS(t-s)\alpha(s,r(s))ds + \sum_{i=1}^{d}
 \int_0^tS(t-s)\sigma_i(s,r(s))dW_i(s), \quad\text{a.s.}
\end{equation}

Let $P(t,T)$ be the time-$t$ price of a zero-coupon bond maturing at $T\ge t$.
We assume that  
\begin{equation*}
 P(t,T)=\exp\left(-\int_0^{T-t}r(t,x)dx\right), \quad 0\le t\le T<\infty.
\end{equation*}
Then the process $f(t,T):=r(t,T-t)$, $0\le t\le T<\infty$, 
satisfies 
\[
 f(t,T)= -\frac{\partial}{\partial T}\log P(t,T), \quad 0\le t\le T<\infty,  
\]
and so is interpreted as the forward rate process. 
If we set by abuse of notation 
$\sigma_i(t,T,\omega)=\sigma_i(t,\omega,r(t))(T-t)$ and 
$\alpha(t,T,\omega)=\alpha(t,\omega,r(t))(T-t)$, then by \eqref{eq:2.4}, 
\begin{align*}
 f(t,T) &= S(t)r_0(T-t) + \int_0^tS(t-s)\alpha(s,s+T-t)ds \\ 
  &\quad + \sum_{i=1}^d\int_0^t S(t-s)\sigma_i(s,s+T-t)dW_i(s) \\ 
  &= r_0(T) + \int_0^t\alpha(s,T)ds + \sum_{i=1}^d\int_0^t\sigma_i(s,T)dW_i(s). 
\end{align*}
Thus $f$ is an HJM forward-rate model.
Further, let $\{B(t)\}_{t\in\mathbb{R}_+}$ be the bank account process defined by 
\begin{equation*}
 B(t) = \exp\left(\int_0^t r(s,0)ds\right), \quad t\in\mathbb{R}_+. 
\end{equation*}
Under the standard integrability conditions for the HJM drift restriction,
the definition of $\alpha$ implies that $\mathbb{P}$ is a local martingale
measure; that is,
the process $\{P(t,T)/B(t)\}_{0\le t\le T}$ is a local martingale under $\mathbb{P}$ for any 
$T>0$. 
Consequently, the infinite-dimensional SDE \eqref{eq:2.2} leads to
a risk-neutral modeling of interest rate processes.

\section{Collocation methods for HJMM equations}\label{sec:3}

In this section, we describe an approximation method for the equation \eqref{eq:2.2} based on 
the kernel-based interpolation theory, and derive its error bound.  

Let $\Phi: \mathbb{R}\to \mathbb{R}$ be a radial and 
positive definite function, 
i.e., $\Phi(\cdot)=\phi(|\cdot|)$ for some $\phi:[0,\infty)\to\mathbb{R}$ and 
for every $\ell\in\mathbb{N}$, for all pairwise distinct 
$y_1,\ldots, y_{\ell}\in\mathbb{R}$ and for all 
$\alpha=(\alpha_i)\in\mathbb{R}^{\ell}\setminus\{0\}$, we have 
\begin{equation*}
  \sum_{i,j=1}^{\ell}\alpha_i\alpha_j\Phi(y_i-y_j)>0. 
\end{equation*}
By Theorems 10.10 and 10.11 in Wendland \cite{wen:2005},
there exists a unique Hilbert space $\mathcal{N}_{\Phi}(\mathbb{R})$ with 
norm $\|\cdot\|_{\mathcal{N}_{\Phi}(\mathbb{R})}$, called the 
native space, of real-valued functions on $\mathbb{R}$ such that 
$\Phi$ is a reproducing kernel for $\mathcal{N}_{\Phi}(\mathbb{R})$. 

Let $\Gamma=\{x_1,\cdots,x_N\}$ be a finite subset of $(0,\infty)$ 
such that $0<x_1<\cdots <x_N$ and 
put $K=\{\Phi(x_i-x_j)\}_{1\le i,j\le N}$. 
Then $K$ is invertible; hence, for any $g:\mathbb{R}_+\to\mathbb{R}$, the function
\begin{equation*}
I(g)(x) = \sum_{j=1}^N(K^{-1}g|_{\Gamma})_j\Phi(x-x_j), \quad x\in\mathbb{R}_+, 
\end{equation*}
interpolates $g$ on $\Gamma$. 

We use a Wendland kernel for $\Phi$.  For a given
$\tau\in\mathbb{N}\cup\{0\}$, set
$\Phi_{\tau}(x)=\phi_{\tau}(|x|)$, $x\in\mathbb{R}$, where
\begin{equation*}
 \phi_{\tau}(r)=\int_r^{\infty}r_{\tau}\int_{r_{\tau}}^{\infty}r_{\tau-1}\int_{r_{\tau-1}}^{\infty}
 \cdots\: r_2\int_{r_2}^{\infty}r_1\max\{1-r_1,0\}^{\nu}dr_1dr_2\cdots dr_{\tau}, \quad 
 r\ge 0 
\end{equation*}
for $\tau\ge 1$, while $\phi_0(r)=\max\{1-r,0\}$, with
$\nu=\tau+1$ in the nested-integral formula.
Theorems 9.12 and 9.13 in \cite{wen:2005} show that
the function $\phi_{\tau}$ is represented as 
\[
 \phi_{\tau}(r)=
 \begin{cases}
  p_{\tau}(r), & 0\le r\le 1, \\
  0, & r>1, 
 \end{cases}
\]
where $p_{\tau}$ is a univariate polynomial with degree $\nu + 2\tau$ having representation 
\begin{equation}
\label{eq:3.1}
 p_{\tau}(r)=\sum_{j=0}^{\nu+2\tau}d_{j,\tau}^{(\nu)}r^j. 
\end{equation}
The coefficients in \eqref{eq:3.1} are given by 
\begin{align*}
 d_{j,0}^{(\nu)} &= (-1)^j\frac{\nu!}{j!(\nu-j)!}, \quad 0\le j\le \nu, \\
 d_{0,s+1}^{(\nu)} &= \sum_{j=0}^{\nu+2s}\frac{d_{j,s}^{(\nu)}}{j+2}, \quad 
 d_{1,s+1}^{(\nu)} =0, \quad s\ge 0, \\
 d_{j,s+1}^{(\nu)} &= -\frac{d_{j-2,s}^{(\nu)}}{j}, \quad s\ge 0, \quad 2\le j\le\nu+2s+2, 
\end{align*} 
in a recursive way for $0\le s\le \tau-1$.  
Further, it is known that 
\begin{equation*}
 \phi_{\tau}(r)\doteq  
\begin{cases}
 \displaystyle\int_r^1 s(1-s)^{\tau+2}(s^2-r^2)^{\tau-1}ds, & 0\le r\le 1, \\[1em]
 0, & r>1, 
\end{cases}
\end{equation*}
where $\doteq$ denotes equality up to a positive constant factor 
(see Chernih et al.~\cite{che-slo-wom:2014}). 
For example, 
\begin{align*}
 \phi_{2}(r)&\doteq \max\{1-r,0\}^5(8r^2+5r+1), \\
 \phi_{3}(r)&\doteq \max\{1-r,0\}^7(21r^3 + 19r^2 + 7r + 1),  \\
 \phi_{4}(r) &\doteq  \max\{1-r,0\}^9(384r^4 + 453r^3 + 237r^2 + 63r + 7). 
\end{align*}

The function $\Phi_{\tau}$ is $C^{2\tau}$ on $\mathbb{R}$, and 
the native space $\mathcal{N}_{\Phi_{\tau}}(\mathbb{R})$ 
coincides with $H^{\tau+1}(\mathbb{R})$ 
where $H^{\theta}(\mathbb{R})$ is the Sobolev space on $\mathbb{R}$ of order 
$\theta\ge 0$ based on $L^2$-norm. 
Further, the native space norm $\|\cdot\|_{\mathcal{N}_{\Phi}(\mathbb{R})}$ and 
the Sobolev norm $\|\cdot\|_{H^{\tau+1}(\mathbb{R})}$ are equivalent.

In what follows, we fix $\tau\in\mathbb{N}$ and $\Phi=\Phi_{\tau}$. 
We also assume that $\Gamma\subset (0,R]$ for some $R>0$.
Since we can expect that 
\begin{gather*}
 r(t,x) \approx I(r(t,\cdot))(x)
 =\sum_{j=1}^N(K^{-1}r(t,\cdot)|_{\Gamma})_{j}\Phi(x-x_{j}), \\ 
 \frac{\partial}{\partial x}r(t,x)\approx
 \sum_{j=1}^N(K^{-1}r(t,\cdot)|_{\Gamma})_j\Phi^{\prime}(x-x_j),
\end{gather*}
this suggests the approximation
\[
 dr(t)\simeq \left\{AI(r(t))+\alpha(t,I(r(t)))\right\}dt 
  + \sum_{i=1}^d \sigma_i(t,I(r(t)))dW_i(t). 
\]
Because the interpolants depend only on values at $\Gamma$, the right-hand
side is finite-dimensional. Let us now assume that $\Gamma$ and $R$ are described by a single parameter
$h>0$. 
Let $0=t_0<t_1<\cdots<t_n=T$ and denote by $r^h(t_k,x)$, $k=0,\ldots,n$, 
$x\in\mathbb{R}_+$, the process obtained by Euler-Maruyama approximation of the 
SDE above, i.e., 
\begin{equation*}
\begin{aligned}
 r^h(t_{k+1},x)&=r^h(t_k,x) + \left\{\frac{d}{dx}I(r^h(t_k))(x)
 + \alpha(t_k,I(r^h(t_k)))(x)\right\}
 \Delta t_{k+1}  \\ 
  &\quad + \sum_{i=1}^d\sigma_i(t_k,I(r^h(t_k)))(x)\Delta W_i(t_{k+1}), 
 \quad k=0,\ldots,n-1, \\
 r^h(t_0,x) &= r_0(x), 
\end{aligned}
\end{equation*}
where $\Delta t_{k+1}=t_{k+1}-t_k$ and $\Delta W_i(t_{k+1})=W_i(t_{k+1})-W_i(t_k)$. 
For $t\in [t_k,t_{k+1}]$ we set  
\begin{align*}
 r^h(t,x) &= r^h(t_k,x) + \left\{\frac{d}{dx}I(r^h(t_k))(x)
 + \alpha(t_k,I(r^h(t_k)))(x)\right\}(t-t_k) \\
  &\quad + \sum_{i=1}^d\sigma_i(t_k,I(r^h(t_k)))(x)(W_i(t)-W_i(t_k)).  
\end{align*}

We also write $I(v)(x)=\sum_{j=1}^N(K^{-1}v)_j\Phi(x-x_j)$
for $v\in\mathbb{R}^N$ by an abuse of notation, and set 
\begin{align*}
 \alpha(t,v) &= (\alpha(t,I(v))(x_1),\ldots,\alpha(t,I(v))(x_N))^{\mathrm{T}}, \\ 
 \sigma(t,v) &=\bigl(\sigma_i(t,I(v))(x_j)\bigr)_
 {\substack{1\le j\le N\\1\le i\le d}}\in\mathbb{R}^{N\times d},
 \quad (t,v)\in [0,T]\times\mathbb{R}^N. 
\end{align*}
Define
\[
 K_1=\{\Phi^{\prime}(x_j-x_{\ell})\}_{1\le j,\ell\le N}.
\]
Then $I^{\prime}(\phi)(x_j)=(K_1K^{-1}\phi|_{\Gamma})_j$.
The vector $r^h_k:=(r^h(t_k,x_1),\ldots,r^h(t_k,x_N))^{\mathsf{T}}\in\mathbb{R}^N$,
$k=0,\ldots,n$, is given by 
\[
 r^h_{k+1} = r^h_k + (K_1K^{-1}r_k^h + \alpha(t_k,r_k^h))\Delta t_{k+1} 
  + \sigma(t_k,r^h_k)\Delta W(t_{k+1})
\]
Starting from
$r^h_0 = (r_0(x_1),\ldots,r_0(x_N))^{\mathsf{T}}$, this recursion is the
Euler--Maruyama approximation of the $N$-dimensional stochastic differential equation
\[
\left\{
\begin{aligned}
 d\tilde{r}(t) &= \left[K_1K^{-1}\tilde{r}(t) + \alpha(t, \tilde{r}(t))\right]dt 
  + \sigma(t,\tilde{r}(t))dW(t), \\
  \tilde{r}(0)&=(r_0(x_1),\ldots,r_0(x_N))^{\mathsf{T}}, 
\end{aligned}
\right.
\]
Suppose, in addition, that
we compute $r^h(t_k,\cdot)$ at points in $\Gamma_e=\{\xi_1,\ldots,\xi_M\}\subset [0,\infty)$. 
The vector $\tilde{r}^h_k:=(r^h(t_k,\xi_1),\ldots,r^h(t_k,\xi_M))^{\mathsf{T}}$,
$k=0,\ldots,n$, is given by 
\[
 \tilde{r}^h_{k+1} = \tilde{r}^h_k + (K_{1e}K^{-1}r_k^h + \tilde{\alpha}(t_k,r_k^h))\Delta t_{k+1} 
  + \tilde{\sigma}(t_k,r^h_k)\Delta W(t_{k+1})
\]
 with $\tilde{r}^h_0=(r_0(\xi_1),\ldots,r_0(\xi_M))^{\mathsf{T}}$, where 
$K_{1e}=\{\Phi^{\prime}(\xi_j-x_{\ell})\}_{1\le j\le M,\;1\le\ell\le N}$, 
\begin{align*}
 \tilde{\alpha}(t,v) &= (\alpha(t,I(v))(\xi_1),\ldots,
   \alpha(t,I(v))(\xi_M))^{\mathrm{T}}, \\ 
 \tilde{\sigma}(t,v) &=\bigl(\sigma_i(t,I(v))(\xi_j)\bigr)_
 {\substack{1\le j\le M\\1\le i\le d}}
  \in\mathbb{R}^{M\times d},
 \quad (t,v)\in [0,T]\times\mathbb{R}^N. 
\end{align*}

The rest of this section is devoted to a conditional convergence analysis of
the approximation above.
To this end, we impose the following conditions on $r_0$ and $\sigma_i$'s: 
\begin{assum}
\label{assum:3.1}
\begin{enumerate}[\rm (i)]
\item The function $r_0$ belongs to $U\cap C_b^2(\mathbb{R}_+)$. 
\item The function $x\mapsto\sigma_i(t,\omega,\phi)(x)$ is $C^2$ on
$\mathbb{R}_+$ for every $i=1,\ldots,d$, $t\in\mathbb{R}_+$,
$\omega\in\Omega$, and $\phi\in U$.
\item There exist 
a nonnegative and Borel measurable function $\Psi$ on $\mathbb{R}_+$ with  
$\Psi^2 w\in L^1(\mathbb{R}_+)$, 
$\lim_{x\to\infty}\Psi(x)=0$ and a positive constant $T_0$ such that 
for $i=1,\ldots,d$, $t,x\in\mathbb{R}_+$, $\omega\in \Omega$,
$\phi\in U$, and $m=0,1,2$,
\[
 \left|\frac{d^m}{dx^m}\sigma_i(t,\omega,\phi)(x)\right|\le \Psi(x), 
\] 
and that for $i=1,\ldots,d$, $t,s,x\in\mathbb{R}_+$, 
$\omega\in \Omega$, $m=0,1,2$, and $\phi,\psi\in U$,  
\[
 \left|\frac{d^m}{dx^m}\sigma_i(t,\omega,\phi)(x) - 
  \frac{d^m}{dx^m}\sigma_i(s,\omega,\psi)(x)\right| 
 \le \Psi(x)\sqrt{|t-s|} + \Psi(x)\int_0^{T_0}|\phi(y)-\psi(y)|dy. 
\]
\end{enumerate}
\end{assum}
Notice that Assumption \ref{assum:2.2} holds under Assumption \ref{assum:3.1} since 
$|\phi(y)-\psi(y)|\le C\|\phi-\psi\|_U$ by Proposition \ref{prop:2.1} (i).  
Thus there exists a unique $U$-valued predictable process $\{r(t)\}_{t\ge 0}$ satisfying 
\eqref{eq:2.3} and \eqref{eq:2.4}. 
Set
\[
 \Delta t =\max_{1\le i\le n}(t_i-t_{i-1}), \quad 
 \Delta x = \sup_{x\in (0,R)}\min_{j=1,\ldots,N}|x-x_j|. 
\]
Because $\Gamma$ and $R$ depend on $h$, so does $\Delta x$.  We also allow
the time grid $\{t_k\}_{k=0}^n$ to depend on $h$, and hence $\Delta t$
depends on $h$ as well.

For $j=1,\ldots,N$, we write $Q_j$ for the cardinal function defined by 
\[
 Q_j(x) = \sum_{i=1}^N(K^{-1})_{ij}\Phi(x-x_i), \quad x\in\mathbb{R}, \;\; j=1,\ldots,N. 
\] 
Set
\[
 \mathcal D_h=[0,T_0]\cup\Gamma\cup\Gamma_e.
\]
\begin{assum}
\label{assum:3.2}
\begin{enumerate}[\rm (i)]
\item The parameters $\Delta t$, $R$, $N$, and $\Delta x$ satisfy 
$\Delta t\to 0$, $R\to \infty$, $N\to\infty$, and $\Delta x\to 0$ as $h\searrow 0$.
Moreover, $\Gamma_e\subset[0,R]$ and $T_0\le R$ for all sufficiently small $h$.
\item There exist $h_0\in(0,1]$ and $C_{\rm stab}>0$ such that
\[
 \sup_{0<h\le h_0}\max_{m=0,1,2}\sup_{x\in\mathcal D_h}
 \sum_{j=1}^N\left|\frac{d^mQ_j}{dx^m}(x)\right|
 \le C_{\rm stab}.
\]
\item The spatial consistency factor satisfies
\[
 (\Delta x)^{(2\tau-1)/\tau}R^{1/\tau}\longrightarrow0
 \qquad\text{as }h\searrow0.
\]
\end{enumerate}
\end{assum}


\begin{rem}
\label{rem:3.2}
Assumption \ref{assum:3.2}(ii) is imposed directly and is not derived from a
geometric counting condition. Suppose, in addition, that $\Gamma$ is
quasi-uniform in the sense that
\[
 c_4 RN^{-1}\le \Delta x\le c_5 RN^{-1} 
\]
for some positive constants $c_4,c_5$.  Then Assumption
\ref{assum:3.2}(iii) follows from
\[
 R=o\bigl(N^{(2\tau-1)/(2\tau)}\bigr).
\]
\end{rem}

\begin{thm}
\label{thm:3.5}
Suppose that Assumptions {\rm \ref{assum:3.1}} and {\rm \ref{assum:3.2}} hold
and that $\tau\ge 3$.
Then there exists $h_0\in (0,1]$ such that 
for every $T\in (0,\infty)$ we have
\[
 \mathbb{E}|r(t,x)-r^h(t,x)|^2\le C\Delta t 
 +C(\Delta x)^{(2\tau-1)/\tau}R^{1/\tau}, \quad x\in\mathcal D_h,
    \;\; 0\le t\le T, \;\; h\le h_0. 
\]
\end{thm}

To prove Theorem \ref{thm:3.5}, we need several preliminary lemmas.
For any $f:\mathbb{R}\to\mathbb{R}$,
\begin{equation*}
 I(f)(x)=\sum_{j=1}^N(K^{-1}f|_{\Gamma})_j\Phi(x-x_j)
 = \sum_{j=1}^Nf(x_j)Q_j(x), \quad x\in\mathbb{R}. 
\end{equation*}
\begin{lem}
\label{lem:3.7}
Suppose that Assumption {\rm \ref{assum:3.2}} and $\tau\ge 3$ hold. 
Let $h_0$ be the constant in Assumption {\rm \ref{assum:3.2}(ii)}. Then for $h\in (0,h_0]$
we have 
\begin{enumerate}[\rm (i)]
\item for $u\in H^{\tau+1}(\mathbb{R})$ 
\[
  \left\|\frac{d}{dx}(u - I(u))\right\|_{L^{\infty}([0,R])}\le C(\Delta x)^{\tau-1}
  \|u\|_{H^{\tau+1}(\mathbb{R})}; 
\]
\item for $m=0,1$ and $u\in C^{1+m}_b(\mathbb{R}_+)$
\[
 \sup_{x\in\mathcal D_h}\left|\frac{d^m}{dx^m}(u-I(u))(x)\right|
 \le C\|u\|_{C^{1+m}(\mathbb{R}_+)}
 (\Delta x)^{(\tau+1/2-m)/(\tau+1-m)}R^{1/(2(\tau+1-m))}.
\]
\end{enumerate}
\end{lem}
\begin{proof}
We first record why the constant in the zeros estimate can be chosen
independently of the expanding radius $R$.  Translate $[0,R]$ to
$(-R/2,R/2)$ and then scale this interval to $(-1,1)$.  The scaled fill
distance is $2\Delta x/R$, which tends to zero.  The fixed interval
$(-1,1)$ satisfies a uniform interior cone condition, so
\cite[Theorem 11.32 and Corollary 11.33]{wen:2005}, with derivative order
one, gives
\[
 \|(u-I(u))'\|_{L^\infty([0,R])}
 \le C(\Delta x)^{\tau-1/2}
       \|u-I(u)\|_{H^{\tau+1}(\mathbb R)}.
\]
Kernel interpolation is the native-space orthogonal projection.  Moreover,
$\mathcal N_{\Phi_\tau}(\mathbb R)=H^{\tau+1}(\mathbb R)$ with equivalent
norms.  Consequently,
\[
 \|u-I(u)\|_{H^{\tau+1}(\mathbb R)}
 \le C\|u\|_{H^{\tau+1}(\mathbb R)}.
\]
For sufficiently small $h$, $\Delta x\le1$, and claim (i) follows after
weakening $(\Delta x)^{\tau-1/2}$ to $(\Delta x)^{\tau-1}$.

To show the claim (ii), let $u\in C^{1+m}_b(\mathbb{R}_+)$. 
We define an extension $\tilde{u}$ on $\mathbb{R}$ of $u$ by 
\[
 \tilde{u}(x) = 
 \begin{cases}
  u(x)\zeta(x), & x\ge 0, \\
  \bigl(u(0)+u^{\prime}(0)x
   +(m/2)(d^{m}u^{\prime}/dx^{m})(0)x^2\bigr)\zeta(x), & x<0,
 \end{cases}
\]
where $\zeta$ is $C^{\infty}$-function on $\mathbb{R}$ such that $0\le \zeta\le 1$, 
$\zeta = 1$ on $(-\delta,\infty)$, and $\zeta = 0$ on $(-\infty,-2\delta)$, for some 
fixed $\delta>0$. 
It follows directly that
$\tilde{u}$ is $C^{1+m}_b(\mathbb{R})$ such that 
$d^{\kappa}\tilde{u}/dx^{\kappa}=d^{\kappa}u/dx^{\kappa}$ on $\mathbb{R}_+$ 
for $0\le \kappa\le 1+m$ and 
$\|\tilde{u}\|_{C^{1+m}(\mathbb{R})}\le C\|u\|_{C^{1+m}(\mathbb{R}_+)}$. 
Next, take a compactly supported $C^{\infty}$ function $\rho$ with unit integral,
and set $\rho_{\varepsilon}(x)=(1/\varepsilon)\rho(x/\varepsilon)$ for 
$x\in\mathbb{R}$ and $\varepsilon>0$. 
With this mollifier and the function $\tilde{u}$, we define $u_{\varepsilon}$ by 
\[
 u_{\varepsilon}(x) = \int_{-\infty}^{\infty}\tilde{u}(y)\rho_{\varepsilon}(x-y)dy, 
 \quad x\in\mathbb{R}. 
\]
This function satisfies 
\begin{equation}
\label{eq:3.2} 
\begin{aligned}
 &\sup_{x\in\mathbb{R}}\left|\frac{d^m}{dx^m}(\tilde{u}(x)-u_{\varepsilon})(x)\right|
   \le C\|u\|_{C^{1+m}(\mathbb{R}_+)}\varepsilon, \\ 
 &\sup_{x\in\mathbb{R}}\left|\frac{d^{\kappa}}{dx^{\kappa}}u_{\varepsilon}(x)\right| 
 \le C\varepsilon^{-\max\{\kappa-1-m,0\}}\|u\|_{C^{1+m}(\mathbb{R}_+)}, 
 \quad \kappa\in\mathbb{N}\cup\{0\}. 
\end{aligned}
\end{equation}
We take another $C^{\infty}$-function $\zeta_1$ on $\mathbb{R}$ such that 
$0\le\zeta_1\le 1$ on $\mathbb{R}$, $\zeta_1(x)=1$ for $|x|\le 1$, and 
$\zeta_1(x)=0$ for $|x|>1+c$ for some $c>0$.  
Then consider the function $\hat{u}_{\varepsilon}(x):=u_{\varepsilon}(x)\zeta_1(x/R)$, 
$x\in\mathbb{R}$. By construction, $\hat{u}_{\varepsilon}\in H^{\tau+1}(\mathbb{R})$, and
by \eqref{eq:3.2}, 
\[
 \|\hat{u}_{\varepsilon}\|_{H^{\tau+1}(\mathbb{R})}^2
 \le C\sum_{\kappa=0}^{\tau+1}\int_{-(1+c)R}^{(1+c)R}
 \left|\frac{d^{\kappa}}{dx^{\kappa}}u_{\varepsilon}(x)
  \right|^2dx 
 \le CR\varepsilon^{-2(\tau-m)}\|u\|^2_{C^{1+m}(\mathbb{R}_+)}. 
\]
Using these estimates and applying the same zeros estimate, now with
derivative order $m$, to $\hat{u}_{\varepsilon}$,
we have 
\begin{align*}
 \sup_{0\le x\le R}\left|\frac{d^m}{dx^m}(\hat{u}_{\varepsilon}-I(\hat{u}_{\varepsilon}))(x)\right| 
 &\le C(\Delta x)^{\tau+1/2-m}\|\hat{u}_{\varepsilon}\|_{H^{\tau+1}(\mathbb{R})} \\ 
 &\le C(\Delta x)^{\tau+1/2-m}R^{1/2}\varepsilon^{-(\tau-m)}
 \|u\|_{C^{1+m}(\mathbb{R}_+)}.
\end{align*}
This together with Assumption \ref{assum:3.2}(ii) leads to
\begin{align*}
 &\left|\frac{d^m}{dx^m}(u-I(u))(x)\right| \\ 
 &\le \left|\frac{d^m}{dx^m}(\tilde{u}-u_{\varepsilon})(x)\right| 
  + \left|\frac{d^m}{dx^m}(\hat{u}_{\varepsilon}-I(\hat{u}_{\varepsilon}))(x)\right| 
  + \left|\frac{d^m}{dx^m}(I(u_{\varepsilon})-I(\tilde{u}))(x)\right| \\ 
 &\le C\|u\|_{C^{1+m}(\mathbb{R}_+)}\varepsilon
    + C(\Delta x)^{\tau+1/2-m}R^{1/2}\varepsilon^{-(\tau-m)}
      \|u\|_{C^{1+m}(\mathbb{R}_+)}
\end{align*}
for all $x\in\mathcal D_h$.
Minimizing the preceding upper bound over $\varepsilon>0$,
we obtain the claim (ii). 
\end{proof}

\begin{lem}
\label{lem:3.8}
Suppose that the assumptions imposed in Theorem {\rm \ref{thm:3.5}} hold. 
Let $h_0$ be the constant in Assumption {\rm \ref{assum:3.2}(ii)}. Then,
\[
 \sup_{0<h\le h_0}\sup_{x\in\mathcal D_h}\max_{k=0,\ldots,n}
  \mathbb{E}|r^h(t_k,x)|^2 < \infty. 
\]
\end{lem}
\begin{proof}
After decreasing $h_0$ if necessary, assume that $\Delta t\le1$.
Fix $k=0,1,\ldots,n-1$ and $x\in\mathcal D_h$. We use the representation
\begin{align*}
 |r^h(t_{k+1},x)|^2 &= |r^h(t_k,x)|^2 + \Lambda(t_k,x)^2(\Delta t_{k+1})^2 + 
 \left(\sum_{i=1}^d\Theta_i(t_k,x)\Delta W_i(t_{k+1})\right)^2 \\ 
 &\quad + 2r^h(t_k,x)\Lambda(t_k,x)\Delta t_{k+1} + 2\Lambda(t_k,x)\Delta t_{k+1}
 \sum_{i=1}^d\Theta_i(t_k,x)\Delta W_i(t_{k+1}) \\ 
 &\quad + 2r^h(t_k,x)\sum_{i=1}^d\Theta_i(t_k,x)\Delta W_i(t_{k+1}), 
\end{align*}
where for $i=1,\ldots,d$ and $x\in\mathbb{R}_+$ 
\begin{align*}
 \Lambda(t_k,x) &= \frac{d}{dx}I(r^h(t_k))(x)
  + \alpha(t_k,I(r^h(t_k)))(x), \\
 \Theta_i(t_k,x) &= \sigma_i(t_k,I(r^h(t_k)))(x). 
\end{align*}
Using Assumptions \ref{assum:3.2}(ii) and \ref{assum:3.1}, we see 
\begin{equation}
\label{eq:3.3}
 \mathbb{E}[\Lambda(t_k,x)^2(\Delta t_{k+1})^2] 
 + \left|\mathbb{E}r^h(t_k,x)\Lambda(t_k,x)\right|\Delta t_{k+1}
 \le C\left(1+\sup_{y\in\mathcal D_h}\mathbb{E}|r^h(t_k,y)|^2\right)
 \Delta t_{k+1}.
\end{equation}
Since $\Lambda(t_k,x)$ and $\Theta_i(t_k,x)$ are $\mathcal{F}_{t_k}$-measurable, 
for $i=1,\ldots,d$ and $j\neq i$,  
\begin{equation}
\label{eq:3.4}
\begin{aligned}
 \mathbb{E}[\Lambda(t_k,x)\Theta_i(t_k,x)\Delta W_i(t_{k+1})] 
 &=\mathbb{E}[r^h(t_k,x)\Theta_i(t_k,x)\Delta W_i(t_{k+1})] \\
 &=\mathbb{E}[\Theta_i(t_k,x)\Theta_j(t_k,x)\Delta W_i(t_{k+1})\Delta W_j(t_{k+1})] \\ 
 &=0. 
\end{aligned}
\end{equation}
We also obtain
\[
 \mathbb{E}[\Theta_i(t_k,x)^2(\Delta W_i(t_{k+1}))^2] 
 =\mathbb{E}[\Theta_i(t_k,x)^2]\Delta t_{k+1}
 \le C\Delta t_{k+1}, \quad i=1,\ldots,d.
\]
From this, \eqref{eq:3.3} and \eqref{eq:3.4} we deduce that 
\[
 \mathbb{E}|r^h(t_{k+1},x)|^2 
 \le (1+C\Delta t_{k+1})\sup_{y\in\mathcal D_h}
 \mathbb{E}|r^h(t_k,y)|^2 + C\Delta t_{k+1},
\]
which leads to the required result. 
\end{proof}

We denote $\tilde{A}=AI$, i.e., $\tilde{A}\phi(x) =I^{\prime}(\phi)(x)$ for $\phi\in U$ and
$x\in\mathbb{R}_+$. Since $\Phi$ is supported in the unit ball,
\begin{align*}
 \|\tilde{A}\phi\|^2_U &= |I^{\prime}(\phi)(0)|^2 + |I^{\prime\prime}(\phi)(0)|^2  
 + \int_0^{R+1}\left\{|I^{\prime\prime}(\phi)(x)|^2
 +|I^{\prime\prime\prime}(\phi)(x)|^2\right\}w(x)dx \\
 &\le C_h\max_{j=1,\ldots,N}|\phi(x_j)|^2 
 \le C_h\|\phi\|_U^2 
\end{align*}
for some positive constant $C_h$ depending on $h$. 
Thus $\tilde{A}:U\to U$ is bounded and generates a uniformly continuous
semigroup $\tilde{S}$ on $U$.

\begin{lem}
\label{lem:3.9}
Suppose that Assumption {\rm \ref{assum:3.2}} and $\tau\ge 3$ hold. 
Let $h_0$ be the constant in Assumption {\rm \ref{assum:3.2}(ii)}. Then
for $h\in (0,h_0]$, $T\in (0,\infty)$ and $\phi\in U\cap C^2_b(\mathbb{R}_+)$ we have 
\[
 \sup_{x\in\mathcal D_h}|S(t)\phi(x)-\tilde{S}(t)\phi(x)|\le
 C(\Delta x)^{(\tau-1/2)/\tau}R^{1/(2\tau)}\|\phi\|_{C^2(\mathbb{R}_+)}, 
 \quad 0\le t\le T. 
\]
\end{lem}
\begin{proof}
Let $\phi\in U\cap C^2_b(\mathbb{R}_+)$ be fixed. 
For fixed $x$, the map $t\mapsto S(t)\phi(x)=\phi(t+x)$ is continuously
differentiable.  The map $t\mapsto\tilde S(t)\phi$ is also continuously
differentiable because $\tilde A$ is bounded.  The fundamental theorem of
calculus therefore gives
\[
 S(t)\phi(x) - \tilde{S}(t)\phi(x) 
 = \int_0^t\left\{AS(\tau)\phi(x)-\tilde{A}\tilde{S}(\tau)\phi(x) 
 \right\}d\tau, \quad 0\le t\le T. 
\]
Thus, 
\begin{align*}
 \sup_{x\in\mathcal D_h}|S(t)\phi(x)-\tilde{S}(t)\phi(x)|
 &\le \|\tilde{A}\|_{\mathcal D_h}
   \int_0^t\sup_{x\in\mathcal D_h}|S(\tau)\phi(x)-\tilde{S}(\tau)\phi(x)|d\tau \\
 &\quad + \int_0^t\sup_{x\in\mathcal D_h}|AS(\tau)\phi(x)-\tilde{A}S(\tau)\phi(x)|d\tau,
\end{align*}
where for $\mathcal{S}:U\to U$, 
\begin{align*}
 \|\mathcal{S}\|_{\mathcal D_h}
 = \sup\left\{\sup_{x\in\mathcal D_h}|\mathcal{S}\psi(x)|/
  \sup_{x\in\mathcal D_h}|\psi(x)| :
\psi\in U, \; \sup_{x\in\mathcal D_h}|\psi(x)| > 0\right\}.
\end{align*} 
Assumption \ref{assum:3.2}(ii) and Lemma \ref{lem:3.7} imply that
$\sup_{0<h\le h_0}\|\tilde{A}\|_{\mathcal D_h}$ is finite, and
\begin{align*}
 \sup_{x\in\mathcal D_h}|AS(\tau)\phi(x)-\tilde{A}S(\tau)\phi(x)|
 &\le C(\Delta x)^{(\tau-1/2)/\tau}R^{1/(2\tau)}\|S(\tau)\phi\|_{C^2(\mathbb{R}_+)} \\ 
 &\le C(\Delta x)^{(\tau-1/2)/\tau}R^{1/(2\tau)}\|\phi\|_{C^2(\mathbb{R}_+)}, 
 \quad 0\le\tau\le T. 
\end{align*}
Thus by Gronwall's lemma, the lemma follows. 
\end{proof}

\begin{proof}[Proof of Theorem {\rm \ref{thm:3.5}}]
First notice that $r^h(t,x)$ can be written as 
\begin{align*}
 r^h(t,x) &= r_0(x) +\int_0^t\left\{(\tilde{A}r^h(s))(x)
  +\alpha(s,I(r^h(s)))(x)\right\}ds
  \\
  &\quad + \sum_{i=1}^d\int_0^t\sigma_i(s,I(r^h(s)))(x)dW_i(s)
  \\
  &\quad - \sum_{i=0}^d\int_0^t\Theta_i(s,x)dW_i(s), 
\end{align*}
where $W_0(t)=t$ and 
\begin{align*}
 \Theta_0(s,x)&=\sum_{k=0}^{n-1}\bigl\{
  \frac{d}{dx}I(r^h(s)-r^h(t_k))(x)
  + \alpha(s,I(r^h(s)))(x) \\
 &\hspace{7em}- \alpha(t_k,I(r^h(t_k)))(x)\bigr\}
  1_{(t_k,t_{k+1}]}(s), \\
 \Theta_i(s,x) &= \sum_{k=0}^{n-1}\left(\sigma_i(s,I(r^h(s)))(x)
  -\sigma_i(t_k,I(r^h(t_k)))(x)\right)
  1_{(t_k,t_{k+1}]}(s), \quad i=1,\ldots,d.  
\end{align*}
Consequently, $r^h$ satisfies the corresponding variation-of-constants formula:
\begin{equation}
\begin{aligned}
 r^h(t) &= \tilde{S}(t)r_0 + \int_0^t\tilde{S}(t-s)\left(\alpha(s,I(r^h(s)))-\Theta_0(s)\right)ds \\ 
  &\quad + \sum_{i=1}^d\int_0^t\tilde{S}(t-s)\left(\sigma_i(s,I(r^h(s)))-\Theta_i(s)\right)dW_i(s). 
\end{aligned}
\end{equation}
Thus 
\begin{align*}
 &r(t)-r^h(t) \\ 
 &= (S(t)-\tilde{S}(t))r_0 + \int_0^t\big\{S(t-s)\alpha(s,r(s)) 
  - \tilde{S}(t-s)\alpha(s,I(r^h(s))) + \tilde{S}(t-s)\Theta_0(s)\big\}ds \\ 
 &\quad + \sum_{i=1}^d\int_0^t\big\{S(t-s)\sigma_i(s,r(s))-\tilde{S}(t-s)\sigma_i(s,I(r^h(s))) 
  + \tilde{S}(t-s)\Theta_i(s)\big\}dW_i(s),  
\end{align*}
and so, for a fixed $x\in\mathcal D_h$,
\begin{equation}
\label{eq:3.6}
 \mathbb{E}|r(t,x)-r^h(t,x)|^2 \le C(\mathcal{I}_1+\mathcal{I}_2 +\mathcal{I}_3 
  +\mathcal{I}_4 +\mathcal{I}_5+\mathcal{I}_6), 
\end{equation}
where
\begin{align*}
 \mathcal{I}_1 &=|S(t)r_0(x)-\tilde{S}(t)r_0(x)|^2, \\
 \mathcal{I}_2 &=\mathbb{E}\int_0^t\left|S(t-s)\alpha(s,r(s))(x)
   -\tilde{S}(t-s)\alpha(s,r(s))(x)\right|^2ds, \\ 
 \mathcal{I}_3 &=\mathbb{E}\int_0^t\left|\tilde{S}(t-s)\alpha(s,r(s))(x)
   -\tilde{S}(t-s)\alpha(s,I(r(s)))(x)\right|^2ds, \\ 
 \mathcal{I}_4 & =\mathbb{E}\int_0^t\left|\tilde{S}(t-s)\alpha(s,I(r(s)))(x)
   -\tilde{S}(t-s)\alpha(s,I(r^h(s)))(x)\right|^2ds, \\
 \mathcal{I}_5 &= \mathbb{E}\int_0^t\left|\tilde{S}(t-s)\Theta_0(s)\right|^2ds, \\ 
 \mathcal{I}_6 &= \sum_{i=1}^d\mathbb{E}\int_0^t\left|S(t-s)\sigma_i(s,r(s))(x) 
  -\tilde{S}(t-s)\sigma_i(s,I(r^h(s)))(x)+ \tilde{S}(t-s)\Theta_i(s)\right|^2ds. 
\end{align*}
By Lemma \ref{lem:3.9}, 
\begin{equation}
\label{eq:3.7}
 \sup_{y\in\mathcal D_h}|S(t)r_0(y)-\tilde{S}(t)r_0(y)|^2\le
 C(\Delta x)^{(2\tau-1)/\tau}R^{1/\tau}. 
\end{equation}
Using Assumption \ref{assum:3.1}, we observe, for $t,s\in [0,T]$ and $\phi,\psi\in U$,  
\begin{equation}
\label{eq:3.8}
\begin{aligned}
 \|\alpha(t,\phi)\|_{C^2(\mathbb{R}_+)}&\le C, \\ 
 \sup_{y\in\mathbb{R}_+}|\alpha(t,\phi)(y)-\alpha(s,\psi)(y)|&\le C
  \sqrt{|t-s|} + C\int_0^{T_0}|\phi(y)-\psi(y)|dy.
\end{aligned}  
\end{equation}
From $\sup_{0<h\le h_0}\|\tilde{A}\|_{\mathcal D_h}<\infty$,
we find 
\begin{equation}
\label{eq:3.9}
 \sup_{0\le\tau\le T}\sup_{0<h\le h_0}\|\tilde{S}(\tau)\|_{\mathcal D_h}<\infty.
\end{equation}
We will also use the corresponding pointwise mean-square bound.  To justify
it, set $q'(x)=(Q_1'(x),\ldots,Q_N'(x))$ and
$G=(Q_j'(x_\ell))_{1\le\ell,j\le N}$.  For $p\ge1$,
\[
 (\tilde A^p\chi)(x)=q'(x)G^{p-1}\chi|_\Gamma.
\]
Assumption \ref{assum:3.2}(ii) implies that the absolute row sums of the
coefficient vector on the right are bounded by $C_{\rm stab}^p$.
Expanding $\tilde S(t)=e^{t\tilde A}$ and applying Minkowski's inequality
therefore gives, for every square-integrable $U$-valued random field $X$,
\begin{equation}
\label{eq:3.9a}
 \sup_{0\le t\le T}\sup_{0<h\le h_0}\sup_{x\in\mathcal D_h}
 \mathbb E|\tilde S(t)X(x)|^2
 \le C\sup_{x\in\mathcal D_h}\mathbb E|X(x)|^2.
\end{equation}
Also, since $r(s)\in C^1_b(\mathbb{R}_+)$, we can apply Lemma \ref{lem:3.7} to obtain  
\begin{equation}
\label{eq:3.10}
 |r(s)(x)-I(r(s))(x)|\le C\|r(s)\|_U(\Delta x)^{(\tau-1/2)/\tau}R^{1/(2\tau)}, 
 \quad x\in\mathcal D_h, \;\; 0\le s\le T.
\end{equation}
By Assumption \ref{assum:3.2}(ii), Lemma \ref{lem:3.9},
\eqref{eq:3.8}--\eqref{eq:3.10}, \eqref{eq:3.9a}, and \eqref{eq:2.3}, we obtain
\begin{equation}
\label{eq:3.11}
 \mathcal{I}_2 \le C(\Delta x)^{(2\tau-1)/\tau}R^{1/\tau}, 
\end{equation}
\begin{equation}
\label{eq:3.12}
\mathcal{I}_3 \le C \int_0^t\int_0^{T_0}\mathbb{E}
|r(s,y)-I(r(s))(y)|^2 dy ds \le C(\Delta x)^{(2\tau-1)/\tau}R^{1/\tau}, 
\end{equation}
and 
\begin{equation}
\label{eq:3.13}
\begin{aligned}
\mathcal{I}_4&\le C\int_0^t
 \int_0^{T_0}\mathbb{E}|I(r(s)-r^h(s))(y)|^2dyds \\
 &\le C\int_0^t\sup_{y\in\mathcal D_h}\mathbb{E}|r(s,y)-r^h(s,y)|^2 ds.
\end{aligned}
\end{equation}
For $s\in(t_k,t_{k+1}]$, Assumption \ref{assum:3.2}(ii), \eqref{eq:3.8}, and
Minkowski's inequality give
\begin{align*}
 \sup_{x\in\mathcal D_h}\mathbb E|\Theta_0(s,x)|^2
 &\le C(s-t_k)
 +C\sup_{y\in\mathcal D_h}
   \mathbb E|r^h(s,y)-r^h(t_k,y)|^2.
\end{align*}
Here the integral over $[0,T_0]$ in \eqref{eq:3.8} is controlled using
the $m=0$ row-sum bound, while the derivative term is controlled using the
$m=1$ bound.  Hence \eqref{eq:3.9a} yields
\begin{align*}
 \mathcal I_5
 &\le C\Delta t+C\sum_{k=0}^{n-1}\Delta t_{k+1}
   \sup_{t_k\le s\le t_{k+1}}
   \sup_{y\in\mathcal D_h}
   \mathbb E|r^h(s,y)-r^h(t_k,y)|^2.
\end{align*}
Furthermore, Assumptions \ref{assum:3.2}(ii) and \ref{assum:3.1} imply
that, for $k=0,\ldots,n-1$, $s\in[t_k,t_{k+1}]$, and
$y\in\mathcal D_h$,
\begin{align*}
 &\mathbb{E}|r^h(s,y)-r^h(t_k,y)|^2 \\ 
 &\le 2\mathbb{E}|I^{\prime}(r^h(t_k))(y) 
  + \alpha(t_k,I(r^h(t_k)))(y)|^2(s-t_k)^2
  +2\sum_{i=1}^d\mathbb{E}|\sigma_i(t_k,I(r^h(t_k)))(y)|^2(s-t_k) \\
 &\le C\left(1+\max_{j=1,\ldots,N}\mathbb{E}|r^h(t_k,x_j)|^2\right)
 (s-t_k).
\end{align*}
From this and Lemma \ref{lem:3.8}, we deduce 
\begin{equation}
\label{eq:3.14}
 \mathcal{I}_5\le C\Delta t.
\end{equation}
For $s\in(t_k,t_{k+1}]$, the integrand defining $\mathcal I_6$ can be
decomposed as
\begin{align*}
 &(S(t-s)-\tilde S(t-s))\sigma_i(s,r(s)) \\
 &\quad+\tilde S(t-s)\{\sigma_i(s,r(s))-\sigma_i(s,I(r(s)))\}\\
 &\quad+\tilde S(t-s)\{\sigma_i(s,I(r(s)))-\sigma_i(s,I(r^h(s)))\}\\
 &\quad+\tilde S(t-s)\{\sigma_i(s,I(r^h(s)))
                    -\sigma_i(t_k,I(r^h(t_k)))\}.
\end{align*}
Lemma \ref{lem:3.9} controls the first term, Lemma \ref{lem:3.7}
and Assumption \ref{assum:3.1} control the second, and
\eqref{eq:3.9a} controls the last two.  The preceding within-step estimate
then gives
\begin{equation}
\label{eq:3.15}
 \mathcal{I}_6\le C\Delta t + C(\Delta x)^{(2\tau-1)/\tau}R^{1/\tau} 
  + C\int_0^t\sup_{y\in\mathcal D_h}
 \mathbb{E}|r(s,y)-r^h(s,y)|^2ds. 
\end{equation}
Consequently, \eqref{eq:3.6}, \eqref{eq:3.7}, \eqref{eq:3.11}--\eqref{eq:3.15} yield 
\[
 \sup_{x\in\mathcal D_h}\mathbb{E}|r(t,x)-r^h(t,x)|^2\le
 C\Delta t + C(\Delta x)^{(2\tau-1)/\tau}R^{1/\tau} + C\int_0^t\sup_{x\in\mathcal D_h}
 \mathbb{E}|r(s,x)-r^h(s,x)|^2ds.
\]
Finally, applying Gronwall's lemma for the function 
$t\mapsto \sup_{x\in\mathcal D_h}\mathbb{E}|r(t,x)-r^h(t,x)|^2$,
we complete the proof of the theorem. 
\end{proof}

\section{Finite-resolution illustrations}\label{sec:4}

This section records two finite-resolution illustrations of the collocation
method.  They do not verify the uniform row-sum condition in Assumption
\ref{assum:3.2}(ii) or the asymptotic spatial balance condition in Assumption
\ref{assum:3.2}(iii).
\begin{ex}[Vasicek model]
The HJMM equation corresponding to the Vasicek model is
\[
 dr(t,x) = \left(\frac{\partial}{\partial x} r(t,x) + \sigma(x)\int_0^x\sigma(y)dy\right)dt 
 + \sigma(x)dW_1(t)
\]
where $\sigma(x)=\sigma e^{-\lambda x}$, $\sigma, \lambda>0$, and 
\[
 r(0,x) = r_0e^{-\lambda x} + b(1-e^{-\lambda x}) 
  - \frac{\sigma^2}{2\lambda^2}(1-e^{-\lambda x})^2. 
\]
The unique mild solution $r(t,x)$ of this equation is given by
\begin{align*}
 r(t,x) &= r(0,t+x) + \int_0^t\sigma(t-s+x)\int_0^{t-s+x}\sigma(y)dyds 
  + \int_0^t\sigma(t-s+x)dW_1(s) \\ 
 &= r(0,t+x) + \frac{\sigma^2}{2\lambda^2}
 \left\{2e^{-\lambda x}(1-e^{-\lambda t})
  - e^{-2\lambda x}(1-e^{-2\lambda t})\right\} \\
 &\quad + \sigma \int_0^te^{-\lambda (t-s+x)}dW_1(s). 
\end{align*}
We apply the collocation method to this model with the following parameters:
$T=1$, $\sigma=0.1$, $b=0.02$, $\lambda=1.0$, and $r_0=0.02$. 
We use the Wendland kernel $\phi_{4}$ with a fixed positive scale parameter.
The time grid is uniform on $[0,1]$, and $\Gamma$ consists of uniform
spatial grid points on $[R/(N+1),R-R/(N+1)]$, where
$R = (1/5)N^{(1-1/(2\tau-2))}$.  The evaluation set
$\Gamma_e=\{\xi_1,\ldots,\xi_{100}\}$ is generated by a Sobol' sequence on
$[R/4,3R/4]$.

This model satisfies Assumption \ref{assum:3.1} with
$w(x)=e^{\lambda x}$ and $\Psi(x)=\sigma(1+\lambda +\lambda^2)e^{-\lambda x}$. 

To assess average finite-resolution performance, we compute the root mean
squared error over $10000$ samples:
\begin{align*}
\text{RMSE} &:= \sqrt{\frac{1}{10000\times 100(n+1)}\sum_{i=0}^n\sum_{j=1}^{100}
 \sum_{\ell=1}^{10000}
|r_{\ell}(t_i,\xi_j)-r^{h}_{\ell}(t_i,\xi_j)|^2},    
\end{align*} 
for several values of $N$ and $n$. 
Here, $r_{\ell}$ and $r^{h}_{\ell}$ are, respectively, the exact and
approximate solutions in the $\ell$th sample.

Table \ref{table:4.1} reports the resulting RMSEs.  Their dependence on the
resolution parameters is nonmonotone; in particular, the table is a performance
record rather than evidence for Theorem \ref{thm:3.5}.
\begin{table}[htb]
\centering
\begin{tabular}[t]{ccccccccc} 
\toprule
& $N$ &   &  $R$ & & $n$  & &$\text{RMSE}$ & \\  \midrule
& $2^6$ & &6.4000 & &$2^{4}$  & & 0.0180 &\\  
&           &  &           & &$2^{6}$  & & 0.0343 & \\  
&           &  &           & &$2^{8}$ & & 0.0683 &\\ \midrule
& $2^7$ & &11.4035 & &$2^{4}$ & & 0.0037 &\\ 
&           & &              & &$2^{6}$ & & 0.0074 &\\ 
&           &  &             & &$2^{8}$ & & 0.0146 &\\ \midrule 
& $2^8$  & & 20.3187 & &$2^{4}$ & & 2.9578e-04 &\\
&            & &               & &$2^{6}$ & & 5.9918e-04 &\\ 
&            & &               & &$2^{8}$ & & 0.0012 &\\ \bottomrule
\end{tabular}
\caption{Root mean squared error for several pairs $(N,n)$.}
\label{table:4.1}
\end{table}
\end{ex}

\begin{ex}[Caplet pricing]
We apply the approximation method to caplet pricing.  Consider the forward
six-month LIBOR prevailing at time
$t$ over the future period $[T,T+0.5]$, defined by 
\[
 1+0.5 L(t,T)=\frac{P(t,T)}{P(t,T+0.5)}, 
\]
or 
\[
 L(t,T)=\frac{1}{0.5}\left(\exp\left(\int_T^{T+0.5}f(t,s)ds\right)-1\right) 
  =\frac{1}{0.5}\left(\exp\left(\int_{T-t}^{T-t+0.5}r(t,x)dx\right)-1\right), 
\]
where $P(t,T)$ and $f(t,s)$ are the price of the discounted bond and 
the forward rate, respectively, given in Section \ref{sec:2}. 
Then the caplet price with $T=10$ is given by 
\begin{equation}
\label{eq:4.1}
 0.5\mathbb{E}\left[e^{-\int_0^{T+0.5}r(s,0)ds}\max(L(T,T)-\kappa,0)\right]. 
\end{equation}
For the volatility functional, we consider the five-year yield-dependent model
\begin{equation*}
 \sigma_1(t,\phi)(x) = \theta_1 e^{-\theta_2 x}
  \max\left(0, \min\left(\frac{1}{5}\int_0^{5}\phi(y)dy, 1\right)\right), 
\end{equation*}
which is addressed in \cite{fil:2001} in a slightly different form. 
As in the previous example, Assumption \ref{assum:3.1} is satisfied with $w(x)=e^{\theta_2x}$ and 
$\Psi(x)=\theta_1(1+\theta_2+\theta_2^2)e^{-\theta_2x}$.

We use the Wendland kernel $\phi_{4}$ with a fixed positive scale parameter.
For the discount-factor quadrature, let
$0=t_0<t_1<\cdots<t_{M_T}=T+0.5$ be a uniform grid, and put
$\delta t=(T+0.5)/M_T$.  The recorded calculation uses spatial grid points
$x_j=j/5$, $j=1,\ldots,50$, on $[R/N,R]$, where
$R = (2/5)N^{\log 25/\log 50}$ so that $R=10$ for $N=50$.

Let $r^h$ be the approximate solution, described in Section \ref{sec:3}, 
of the HJMM equation with the yield-dependent volatility model above. 
We approximate \eqref{eq:4.1} by
\begin{equation}
\label{eq:4.2}
 0.5\mathbb{E}\left[e^{-\delta t\sum_{i=0}^{M_T-1}r^h(t_i,\xi_1)}
  \max(L^h(T,T)-\kappa,0)\right], 
\end{equation}
with 
\[
  L^h(T,T)
  =\frac{1}{0.5}\left(\exp\left(r^h(T,\xi_1)/10+r^h(T,\xi_2)/5+7r^h(T,\xi_3)/40 + r^h(T,\xi_4)/40
  \right)-1\right), 
\]
where $\Gamma_e=\{\xi_j\}$ is the set of the evaluation points given by 
$\xi_j=(j-1)/5$, $j=1,\ldots,51$ and 
we have used the approximation 
\begin{align*}
 \int_0^{0.5}r^h(T,x)dx &\approx \frac{1}{2}(r^h(T,\xi_1)+r^h(T,\xi_2))\times 0.2 
  + \frac{1}{2}(r^h(T,\xi_2)+r^h(T,\xi_3))\times 0.2 \\  
  &\quad +\frac{1}{2}\left(r^h(T,\xi_3) + \frac{1}{2}(r^h(T,\xi_3) + r^h(T,\xi_4))\right)\times 0.1.   
\end{align*}

Figure \ref{fig:4.3} plots the approximated price \eqref{eq:4.2} with $N=50$ for 
$\theta_1\in\{0.05,0.1,\ldots,2.5\}$ and $\theta_2\in\{0.05,\ldots,1.5\}$. 
The computed finite-resolution price surface varies smoothly with
$\theta_1$ and $\theta_2$.  This observation is not an asymptotic stability
test.

\begin{figure}[htbp]
\centering
\includegraphics[width=0.9\columnwidth]{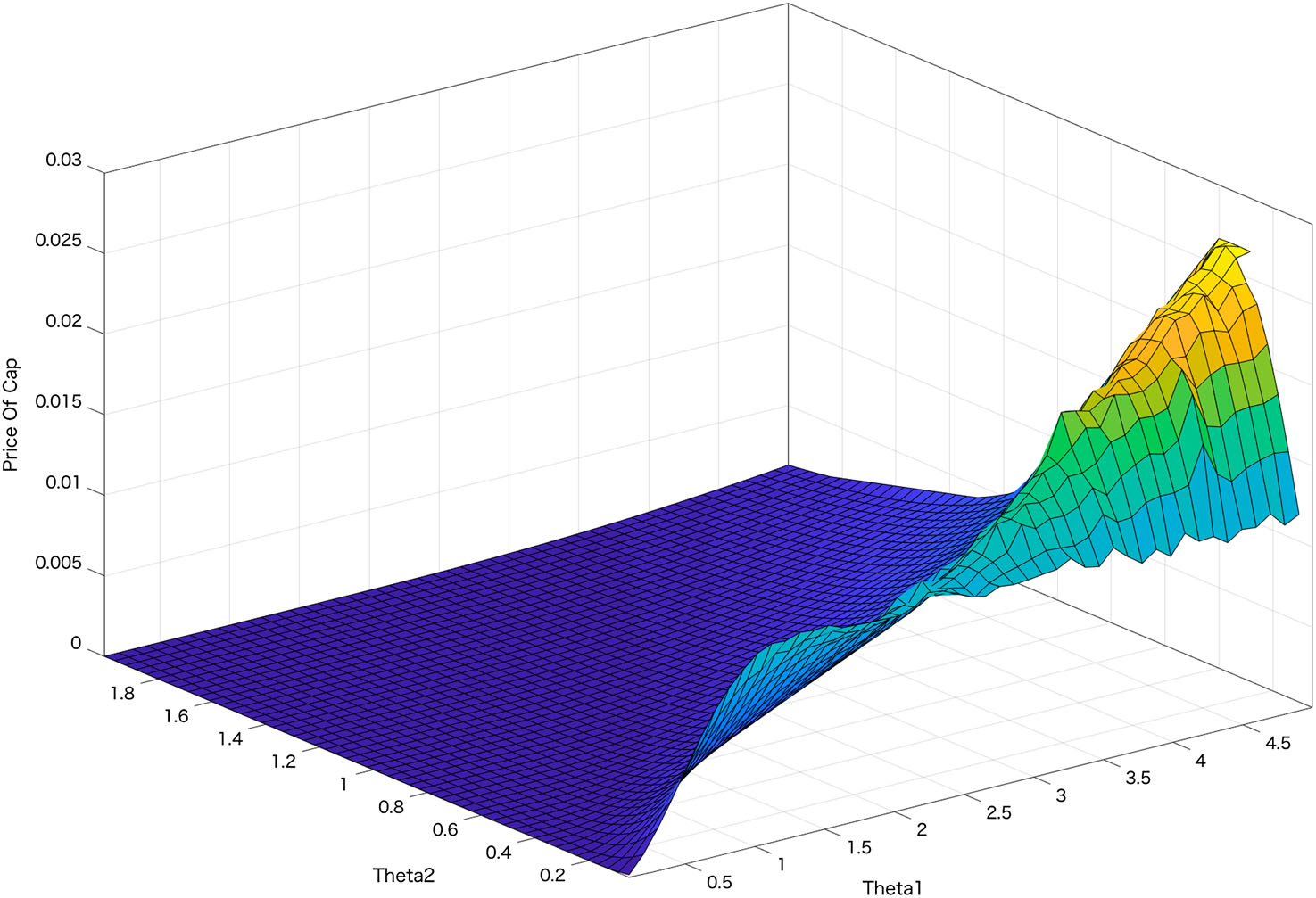}
\caption{Approximated prices for
$T=10$, $\kappa=0.04$, and 
$(\theta_1,\theta_2)\in\{0.05,0.1,\ldots,2.5\}\times \{0.05,\ldots,1.5\}$. }
\label{fig:4.3}
\end{figure}

\end{ex}

\section*{Acknowledgements}

This study was supported by JSPS KAKENHI Grant Number JP17K05359.
The authors thank Kanji Kurihara and Hideki Noda for their assistance with
the numerical experiments in Section \ref{sec:4}.

\begingroup
\hbadness=10000
\bibliographystyle{plain}
\bibliography{../mybib}
\endgroup

\end{document}